\documentclass[letterpaper,11pt]{article}

\usepackage{verbatim, url}
\usepackage{latexsym}
\usepackage{amsmath,amssymb}
\usepackage{epsfig}

\usepackage{fullpage} % [no lncs]

\let\myPushQED=\pushQED
\let\myPopQED=\popQED
\newcommand{\myignore}[1]{}
\newenvironment{proof*}
  {\let\pushQED=\myignore\begin{proof}\let\pushQED=\myPushQED}
  {\def\popQED{}\end{proof}\let\popQED=\myPopQED}

{\makeatletter
 \gdef\xxxmark{%
   \expandafter\ifx\csname @mpargs\endcsname\relax % in minipage?
     \expandafter\ifx\csname @captype\endcsname\relax % in figure/caption?
       \marginpar{xxx}% not in a caption or minipage, can use marginpar
     \else
       xxx % notice trailing space
     \fi
   \else
     xxx % notice trailing space
   \fi}
 \gdef\xxx{\@ifnextchar[\xxx@lab\xxx@nolab}
 \long\gdef\xxx@lab[#1]#2{{\bf [\xxxmark #2 ---{\sc #1}]}}
 \long\gdef\xxx@nolab#1{{\bf [\xxxmark #1]}}
}

\newtheorem{theorem}{Theorem}               % [no lncs] 
\newtheorem{lemma}[theorem]{Lemma}          % [no lncs] 
\newtheorem{corollary}[theorem]{Corollary}    % [no lncs] 
\newtheorem{problem}{Problem}               % [no lncs] 

\newcommand{\fct}{\rightarrow}

\newcommand{\eps}{\varepsilon}

\newcommand{\E}{\textnormal{E}}
\newcommand{\Var}{\textnormal{Var}}

\newcommand{\req}[1]{(\ref{#1})}

\newcommand{\floor}[1]{\lfloor{#1}\rfloor}

\let\phi=\varphi

\def\tsize{t}

\renewcommand{\th}{\ifmmode{^{\textrm{th}}}\else{\textsuperscript{th}\ }\fi}
\newcommand{\nd}{\ifmmode{^{\textrm{nd}}}\else{\textsuperscript{nd}\ }\fi}
\newcommand{\rd}{\ifmmode{^{\textrm{rd}}}\else{\textsuperscript{rd}\ }\fi}

\newcommand{\drop}[1]{}

\newcommand{\qed}{\hbox{\rule{6pt}{6pt}}}
\newenvironment{proof}[1][]{\paragraph{Proof{#1}}}{\hfill\qed\medskip\\}

\title{{\Large Lecture Notes on}\\[1ex]
Linear Probing with 5-Independent Hashing
} 
\author{Mikkel Thorup}

\begin{document}

\maketitle

\begin{abstract}
These lecture notes show that linear probing takes expected constant
time if the hash function is 5-independent. This result was first
proved by Pagh et al.~[STOC'07,SICOMP'09].  The simple proof here is
essentially taken from [P{\v a}tra{\c s}cu and Thorup ICALP'10].  We
will also consider a smaller space version of linear probing that may
have false positives like Bloom filters.  

These lecture notes illustrate the use of higher moments in data
structures, and could be used in a course on randomized algorithms.
\end{abstract}

\section{$k$-independence}
The concept of $k$-independence was introduced by Wegman and
Carter~\cite{wegman81kwise} in FOCS'79 and has been the cornerstone of
our understanding of hash functions ever since. A hash function is a
random function $h: [u] \to [\tsize]$ mapping {\em keys\/} to {\em
  hash values}.  Here $[s]=\{0,\ldots,s-1\}$.  We can also think of a
$h$ as a random variable distributed
over $[\tsize]^{[u]}$.   We say that $h$ is {\em $k$-independent\/} if
for any distinct keys $x_0, \dots, x_{k-1} \in [u]$ and (possibly
non-distinct) hash values $y_0,\ldots, y_{k-1}\in [t]$, we have
$\Pr[h(x_0)=y_0\wedge\cdots \wedge
  h(x_{k-1})=y_{k-1}]=1/t^k$. Equivalently, we can define
$k$-independence via two separate conditions; namely,
\begin{itemize}
\item[(a)] for any distinct keys $x_0, \dots, x_{k-1} \in [u]$, the hash values
$h(x_0), \dots, h(x_{k-1})$ are independent random variables, that
is, for any (possibly non-distinct) hash values 
$y_0,\ldots, y_{k-1}\in [t]$ and
$i\in[k]$, $\Pr[h(x_i)=y_i]=\Pr\left[h(x_i)=y_i\mid\bigwedge_{j\in[k]\setminus\{i\}}h(x_j)=y_j\right]$, and
\item[(b)] for
any $x\in [u]$, $h(x)$ is uniformly distributed in $[\tsize]$. 
\end{itemize}
As the concept of independence is fundamental to probabilistic
analysis, $k$-independent hash functions are both natural and powerful in
algorithm analysis. They allow us to replace the heuristic assumption
of truly random hash functions that are uniformly distributed in $[\tsize]^{[u]}$, hence
needing $u\lg \tsize$ random bits ($\lg=\log_2$), with real implementable hash
functions that are still ``independent enough'' to yield provable
performance guarantees similar to those proved with true randomness. We are then left with the natural goal of
understanding the independence required by algorithms. 

Once we have proved that $k$-independence suffices for a hashing-based 
randomized algorithm, we are free to use {\em any\/} $k$-independent hash function.
The canonical construction of a $k$-independent hash function is based on
polynomials of degree $k-1$. Let $p \ge u$ be prime. Picking random
$a_0, \dots, a_{k-1} \in [p]=\{0, \dots, p-1\}$, the hash function is
defined by:
\begin{equation}\label{eq:prime}
h(x) = \big( a_{k-1} x^{k-1} + \cdots + a_1 x + a_0 \big)
        \bmod{p} 
\end{equation}
If we want to limit the range of hash values to $[t]$, we use
$h'(x)=h(x)\bmod t$. This preserves requirement (a) of independence among
$k$ hash values. Requirement (b) of uniformity is close to 
satisfied if $p\gg t$. More precisely, for any key $x\in [p]$ and
hash value $y\in [t]$, we get $1/t-1/p<\Pr[h'(x)=y]<1/t+1/p$.

Sometimes 2-independence suffices. For example, 2-independence implies
so-called universality \cite{carter77universal}; namely that the probability
of two keys $x$ and $y$ colliding with $h(x)=h(y)$ is $1/\tsize$; or close
to $1/t$ if the uniformity of (b) is only approximate.
Universality implies expected constant time performance of hash tables
implemented with chaining. Universality also suffices
for the 2-level hashing of Fredman et al.
\cite{fredman84dict}, yielding static hash tables with constant query time.

At the other end of the spectrum, when dealing with problems involving $n$ objects, $O(\lg n)$-independence suffices in
a vast majority of applications. One reason for this is the Chernoff
bounds of~\cite{schmidt95chernoff} for $k$-independent events, whose
probability bounds differ from the full-independence Chernoff bound by
$2^{-\Omega(k)}$. Another reason is that random graphs with $O(\lg
n)$-independent edges~\cite{alon08kwise} share many of the properties
of truly random graphs.

The independence measure has long been central to the study of
randomized algorithms. It applies not only to hash functions, but also 
to pseudo-random number generators viewed as assigning hash values to $0,1,2,\ldots$.  For example, \cite{karloff93prg} considers variants of
QuickSort, \cite{alon99linear} consider the maximal bucket size for
hashing with chaining, and   
\cite{cohen09cuckoo5,dietzfelbinger09cuckoo-bas} consider Cuckoo hashing.
In several cases \cite{alon99linear,dietzfelbinger09cuckoo-bas,karloff93prg},
it is proved that linear transformations $x\mapsto \big( (ax + b) \bmod p \big)$ do not suffice for good performance, hence that 
2-independence is not in itself sufficient.

Our focus in these notes is linear probing described below.

\section{Linear probing}\label{sec:linpb}
Linear probing is a classic implementation of hash tables. It uses a
hash function $h$ to map a dynamic set $S$ of keys into an array $T$ of size
$\tsize>|S|$. The entries of $T$ are keys, but we can also see if
an entry is ``empty''. This could be coded, either via an extra bit, or
 via a distinguished nil-key. We start with an empty set $S$ and all empty
locations. When inserting $x$, if the desired location $h(x)\in
[\tsize]$ is already occupied, the algorithm scans $h(x)+1, h(x)+2,
\dots,\tsize-1,0,1,\ldots$ until an empty location is found, and
places $x$ there. Below, for simplicity, we ignore the wrap-around
from $t-1$ to $0$, so a key $x$ is always placed in a location
$i\geq h(x)$.

To search a key $x$, the query algorithm starts at $h(x)$ and scans
either until it finds $x$, or runs into an empty position, which
certifies that $x$ is not in the hash table.  When the query search is
unsuccessful, that is, when $x$ is not stored, the query algorithm
scans exactly the same locations as an insert of $x$. A general bound
on the query time is hence also a bound on the insertion time.

Deletions are slightly more complicated. The invariant we want to
preserve is that if a key $x$ is stored at some location $i\in[t]$,
then all locations from $h(x)$ to $i$ are filled; for otherwise the
above search would not get to $x$. Suppose now that $x$ is deleted
from location $i$. We then scan locations $j=i+1,i+2,\ldots$ for a key $y$
with $h(y)\leq i$. If such a $y$ is found at location $j$, we move $y$ to location
$i$, but then, recursively, we have to try refilling $j$, looking for a later
key $z$ with $h(z)\leq j$. The deletion process terminates when we
reach an empty location $d$, for then the invariant says that there
cannot be a key $y$ at a location $j>d$ with $h(y)\leq d$. The
recursive refillings always visit successive locations, so the
total time spent on deleting $x$ is proportional to the number of
locations from that of $x$ and to the first empty location. Summing up,
we have
\begin{theorem}\label{thm:linpb}
With linear probing, the time it takes to search,
insert, or delete a key $x$ is at most proportional to the number of
locations from $h(x)$ to the first empty location.
\end{theorem}
With $n$ the nunber of keys and $t$ the size of the table, 
we call $n/t$ the {\em load\/} of our table. We generally assume that the load is bounded from $1$, e.g.,~that the number
of keys is $n \le \frac{2}{3}\tsize$. With a good distribution of
keys, we would then hope that the number of locations from $h(x)$ to
an empty location is $O(1)$.

This classic data structure is one of the most popular implementations of hash
tables, due to its unmatched simplicity and efficiency. 
The practical use of linear probing dates back at least to 1954 to an
assembly program by Samuel, Amdahl, Boehme (c.f. \cite{knuth-vol3}).
On modern
architectures, access to memory is done in cache lines (of much more
than a word), so inspecting a few consecutive values is typically
only slightly worse that a single memory access. Even if the scan straddles a
cache line, the behavior will still be better than a second random
memory access on architectures with prefetching. Empirical
evaluations~\cite{black98linprobe,heileman05linprobe,pagh04cuckoo}
confirm the practical advantage of linear probing over other known
schemes, e.g., chaining, but caution~\cite{heileman05linprobe,thorup12kwise} that
it behaves quite unreliably with weak hash functions. Taken together, these findings form a strong
motivation for theoretical analysis.

Linear probing was shown to take expected constant time for
any operation in 1963 by Knuth~\cite{knuth63linprobe}, in a report which is now
regarded as the birth of algorithm analysis. This analysis, however, 
assumed a truly random hash function. 

A central open question of Wegman and Carter~\cite{wegman81kwise} was
how linear probing behaves with $k$-independence. Siegel and
Schmidt~\cite{schmidt90hashing,siegel95hashing} showed that $O(\lg
n)$-independence suffices for any operation to take expected constant time. 
Pagh et al.~\cite{pagh07linprobe} showed that just $5$-independence
suffices for expected constant operation time. 
They also showed that linear transformations do not suffice, hence that 2-independence is not in itself sufficient.

P{\v a}tra{\c s}cu and Thorup \cite{PT16:kwise} proved
that $4$-independence is not in itself sufficient for
expected constant operation time. They display
a concrete combination of keys and a 4-independent random hash function
where searching certain keys takes super constant expected time. This
shows that the $5$-independence result of Pagh et al.~\cite{pagh07linprobe} 
is best possible. In fact \cite{PT16:kwise} provided a complete 
understanding of linear probing with low independence as summarized in
Table~\ref{tab:lin-probe}. 

Considering loads close to $1$, that is load $(1-\eps)$, P{\v a}tra{\c s}cu and Thorup \cite{patrascu12charhash} proved that the expected operation time
is $O(1/\eps^2)$ with 5-independent hashing, matching the bound of
Knuth~\cite{knuth63linprobe} assuming true randomness. 
The analysis from  \cite{patrascu12charhash} 
also works for something called simple tabulation hashing that is
we shall return to in Section \ref{sec:charhash}.

\begin{table}[t]
  \centering
  \begin{tabular}{|l|c|c|c|c|}
    \hline
    Independence & 2 & 3 & 4 & $\ge 5$ \\
    \hline
    Query time & $\Theta(\sqrt{n})$
              & $\Theta(\lg n)$ & $\Theta(\lg n)$ & $\Theta(1)$ \\
    \hline
    Construction time & $\Theta(n\lg n)$ & $\Theta(n\lg n)$ 
              & $\Theta(n)$ & $\Theta(n)$ \\
    \hline
  \end{tabular}
\caption{Expected time bounds for linear probing with a poor $k$-independent
  hash function. The bounds are worst-case expected, e.g., a lower bound
for the query means that there is a concrete combination of stored set, query 
key, and $k$-independent hash function with this expected search time while the upper-bound means that this is the worst expected time for any such combination. Construction time refers to the worst-case expected total time for inserting $n$   keys starting from an empty table.}
\label{tab:lin-probe}
\end{table}

\section{Linear probing with $5$-independence} 
Below we present the simplified version of the proof from
\cite{patrascu12charhash} of the result from \cite{pagh07linprobe}
that 5-independent hashing suffices for expected constant time with
linear probing. For simplicity, we assume that the load is at most 
$\frac 23$. Thus
we study a set $S$ of $n$ keys stored in a linear probing
table of size $t\geq \frac 32 n$. We assume that $t$ is a power of
two.  

A crucial concept is a {\em run\/} $R$ which is a maximal
interval of filled positions. We have an empty position before $R$,
which means that all keys $x\in S$ landing in $R$ must also hash into $R$
in the sense that $h(x)\in R$. Also, we must have exactly $r=|R|$ keys hashing to $R$ since the
position after $R$ is empty.

By Theorem \ref{thm:linpb} the time it takes for any operation on a
key $q$ is at most proportional to the number of locations from $h(x)$
to the first empty location. We upper bound this number by $r+1$ where
$r$ is the length of the run containing $h(q)$. Here $r=0$ if the
location $h(q)$ is empty. We note that the query key $q$ might itself be in
$R$, and hence be part of the run, e.g., in the case of deletions.

We want to give an expected upper bound on $r$. In order to limit the
number of different events leading to a long run, we focus on dyadic
intervals: a {\em (dyadic) $\ell$-interval\/} is an
interval of length $2^\ell$ of the form $[i2^\ell,(i+1)2^\ell)$ where
$i\in[t/2^\ell]$. Assuming that the hashing maps $S$ uniformly into
$[t]$, we expect $n2^\ell/t\leq \frac 23 2^\ell$ keys to hash into a
given $\ell$-interval $I$. We say that $I$ is ``near-full'' if
at least $\frac34\, 2^\ell$ keys from $S\setminus\{q\}$ hash into $I$. 
We claim that
a long run implies that some dyadic interval of similar size is near-full.
More precisely,
\begin{lemma}\label{lem:run} Consider a run $R$ of 
length $r\geq 2^{\ell+2}$. Then one of the first four $\ell$-intervals
intersecting $R$ must be near-full.
\end{lemma}
\begin{proof} Let $I_0,\ldots,I_3$ be the first four $\ell$-intervals
intersecting $R$.
Then $I_0$ may only have its last end-point in $R$ while $I_1,\ldots,I_3$
are contained in $R$ since $r\geq 4\cdot 2^\ell$. 
In particular, this means that
$L=\left(\bigcup_{i\in [4]}I_i\right)\cap R$ has length at least
$3\cdot 2^\ell +1$.

But $L$ is a prefix
of $R$, so all keys landing in $L$ must hash into $L$.
Since $L$ is full, we must have at least $3\cdot 2^\ell +1$
keys hashing into $L$. Even if this includes the query key $q$, then
we conclude that one of our four intervals $I_i$
must have $3\cdot 2^\ell/4\geq \frac 34 2^\ell$ keys from $S\setminus\{q\}$
hashing into it, implying that $I_i$ is near-full.
\end{proof}
Getting back to our original question, we are considering the
run $R$ containing the hash of the query $q$. 
\begin{lemma}\label{lem:query-runs} If the run containing the hash
of the query key $q$ is of length $r\in
[2^{\ell+2},2^{\ell+3})$, then one of the following 12 consecutive 
$\ell$-intervals is near-full: the $\ell$-interval containing $h(q)$, 
the 8 nearest $\ell$-intervals to its left, and the 3 nearest $\ell$-intervals to its right.
\end{lemma}
\begin{proof} Let $R$ be the run containing $h(q)$. 
To apply Lemma \ref{lem:run}, we want to
show that the first four $\ell$-intervals intersecting $R$ has
to be among the 12 mentioned in Lemma \ref{lem:query-runs}. Since the run $R$
containing $h(q)$ has
length less than $8\cdot 2^\ell$, the first $\ell$-interval
intersecting $R$ can be at most 8 before the one containing $h(q)$. The
3 following intervals are then trivially contained among the 12.
\end{proof}
For our analysis, in the random choice of the hash function $h$, 
we first fix the hash value $h(q)$ of the query key $q$. Conditioned on this 
value of $h(q)$, for each $\ell$, let
$P_{\ell}$ be an upper-bound on the the probability that any 
given $\ell$-interval is near-full. Then the probability that
the run containing $h(q)$ has length $r\in [2^{\ell+2},2^{\ell+3})$ is
bounded by $12 P_{\ell}$. Of course, this only gives us a bound for $r\geq 4$.
We thus conclude that the expected length of the run containing the hash of
the query key $q$ is bounded by

Thus, conditioned on the hash of the query key, for each $\ell$ we are
interested in a bound $P_{\ell}$ on the probability that any 
given $\ell$-interval is near-full. Then the probability that
the run containing $h(q)$ has length $r\in [2^{\ell+2},2^{\ell+3})$ is
bounded by $12 P_{\ell}$. Of course, this only gives us a bound for $r\geq 4$.
We thus conclude that the expected length of the run containing the hash of
the query key $q$ is bounded by
\[3+\sum_{\ell=0}^{\log_2 t} 2^{\ell+3}\cdot 12 P_\ell=O\left(1+\sum_{\ell=0}^{\log_2 t}
 2^{\ell} P_\ell\right).\]
Combined with Theorem \ref{thm:linpb}, we have now proved
\begin{theorem}\label{thm:linpb-anal} Consider storing a set
$S$ of keys in a linear probing table of size $t$ where $t$ is
a power of two. Conditioned on the hash of a key $q$, let
$P_{\ell}$ bound the probability that $\frac 34\,2^\ell$ keys
from $S\setminus\{q\}$ hash to any 
given $\ell$-interval. Then the expected time to 
search, insert, or delete $q$ is bounded by
\[O\left(1+\sum_{\ell=0}^{\log_2 t}
 2^{\ell} P_\ell\right).\]
\end{theorem}
We note that Theorem \ref{thm:linpb-anal} does not mention the
size of $S$. However, as mentioned earlier, 
with a uniform distribution, the expected number of elements hashing
to 
an $\ell$-interval is $\leq 2^\ell |S|/t$, so for $P_\ell$ to 
be small, we want this expectation to be significantly smaller than  
$\frac 34\,2^\ell$. Assuming $|S|\leq \frac 23 t$, the expected
number is $\frac 23\,2^\ell$.

To get constant expected cost for linear probing, we are going to
assume that the hash function used is $5$-independent. This means
that no matter the hash value $h(q)$ of $q$, conditioned on $h(q)$,
the keys from $S\setminus \{q\}$ are hashed 4-independently.
This means that if $X_x$ is the indicator variable for
a key $x\in S\setminus\{q\}$ hashing to a given interval $I$,
then the variables $X_x,\ x\in S\setminus\{q\}$ are 4-wise independent.

\subsection{Fourth moment bound}
The probabilistic tool we shall use here to analyze 4-wise independent
variables is a 4\th moment bound.  For
$i\in[n]$, let $X_i\in[2]=\{0,1\}$, $p_i=\Pr[X_i=1]=\E[X_i]$, $X =
\sum_{i\in [n]} X_i$, and $\mu=\E[X]=\sum_{i\in[n]} p_i$.  Also
$\sigma_i^2=\Var[X_i]=\E[(X_i-p_i)^2]=p_i(1-p_i)^2+(1-p_i)p_i^2=p_i-p_i^2$. As
long as the $X_i$ are pairwise independent, the variance of the sum is
the sum of the variances, so we define
\[\sigma^2=\Var[X]=\sum_{i\in[n]}\Var[X_i]=\sum_{i\in [n]}\sigma^2_i\leq \mu.\]
By Chebyshev's inequality, we have 
\begin{equation}\label{eq:2bound}
\Pr[|X-\mu|\geq d\sqrt\mu]\leq \Pr[|X-\mu|\geq d\sigma]\leq 1/d^2.
\end{equation}
We are going to prove a stronger bound if the variables are 4-wise independent
and $\mu\geq 1$ (and which is only stronger if $d\geq 2$).
\begin{theorem}\label{thm:4bound}
If the variables $X_0,\ldots,X_{n-1}\in\{0,1\}$ 
are 4-wise independent, $X=\sum_{i\in [n]} X_i$, and $\mu=\E[X]\geq 1$,
then 
\[\Pr[|X-\mu|\geq d\sqrt\mu]\leq  4/d^4.\]
\end{theorem}
\begin{proof}
Note that $(X-\mu)= \sum_{i\in[n]} (X_i-p_i)$. By linearity of
expectation, the fourth moment is:
\[ \E[(X-\mu)^4] = \E[ \big( \sum_i X_i-p_i \big)^4 ]
= \sum_{i,j,k,l\in[n]} \E\big[ (X_i-p_i)
  (X_j-p_j)(X_k-p_k)(X_l-p_l) \big].
\]
Our goal is to get a good bound on the fourth moment.

Consider a term $\E\big[ (X_i-p_i)
  (X_j-p_j)(X_k-p_k)(X_l-p_l) \big]$. The at most 4 distinct variables 
are completely independent. Suppose one of them, say, $X_i$, appears
only once. By definition, $\E\big[ (X_i-p_i)\big]=0$, and since
it is independent of the other factors, we
get $\E\big[ (X_i-p_i)
  (X_j-p_j)(X_k-p_k)(X_l-p_l) \big]=0$. We can therefore ignore all
terms where any variable appears once. 
We may therefore assume that
each variables appears either twice or 4 times. In terms with
variables appearing twice, we have two indices $a<b$ where $a$ is assigned to two of $i,j,k,l$, while $b$ is assigned to the other two, yielding
$4 \choose 2$ combinations based on $a<b$. Thus
we get 
\begin{align*}\E[(X-\mu)^4]&= \sum_{i,j,k,l\in[n]} \E\big[ (X_i-p_i)
  (X_j-p_j)(X_k-p_k)(X_l-p_l) \big]\\
&= \sum_{i} \E\big[ (X_i-p_i)^4 \big]
+ {4\choose 2}\sum_{a<b}\left( \E\big[ (X_a-p_a)^2\big]
 \E\big[  (X_b-p_b)^2 \big]\right).
\end{align*}
Considering any multiplicity $m=2,3,4,5,\ldots$, we
have 
\begin{equation}\label{eq:vari}
\E[(X_i-p_i)^m]\leq\E[(X_i-p_i)^2]=\sigma_i^2.
\end{equation}
To see this, note that $X_i,p_i\in[0,1]$. Hence
$|X_i-p_i|\leq 1$, so $(X_i-p_i)^{m-2}\leq 1$, and therefore
$(X_i-p_i)^m\leq (X_i-p_i)^2$.
Continuing our calculation, we get
\begin{align}\E[(X-\mu)^4]
&= \sum_{i} \E\big[ (X_i-p_i)^4 \big]
+ {4\choose 2}\sum_{a<b}\left( \E\big[ (X_a-p_a)^2\big]
 \E\big[  (X_b-p_b)^2 \big]\right)\nonumber\\
&\leq \sum_{i} \sigma_i^2
+ {4\choose 2}\sum_{a<b}\sigma_a^2\sigma_b^2\nonumber\\
&\leq \sigma^2
+ 3 \left(\sum_{i} \sigma_i^2\right)^2\nonumber\\
&=\sigma^2
+ 3 \sigma^4.\label{eq:4moment}
\end{align}
Since $\sigma^2\leq\mu$ and $\mu\geq 1$, we
get 
\begin{equation}\label{eq:fourth}
\E[(X-\mu)^4]\leq \mu
+ 3 \mu^2\leq 4\mu^2.
\end{equation}
which is our desired bound on the fourth moment.

by Markov's inequality, 
\[\Pr[|X-\mu|\geq d\sqrt\mu]=\Pr[(X-\mu)^4\geq (d\sqrt\mu)^4]\leq\E[(X-\mu)^4]/(d\sqrt\mu)^4\leq 4/d^4.\]
This completes the proof of Theorem \ref{thm:4bound}.
\end{proof}
We are now ready to prove the 5-independence suffices for
linear probing.
\begin{theorem}\label{thm:linpb-5} Suppose
we use a 5-independent hash function $h$ to store a set $S$ of 
$n$ keys in a linear probing table of size $t\geq \frac 32 n$ where
$t$ is a power of two. Then it takes expected constant time to
search, insert, or delete a key.
\end{theorem}
\begin{proof} First we fix the hash of the query key $q$. 
To apply Theorem \ref{thm:linpb-anal}, we need to find a bound
$P_{\ell}$ on the probability that $\frac 34\,2^\ell$ keys
from $S\setminus\{q\}$ hash to any given $\ell$-interval $I$. For each
key $x\in S\setminus\{q\}$, 
let $X_x$ be the indicator variable for $h(x)\in I$. 
Then $X=\sum_{x\in S\setminus\{q\}} X_x$ is the number
of keys hashing to $I$, and the expectation of $X$ is 
$\mu= \E[X]= n 2^\ell/t\leq \frac 23 2^\ell$.
Our concern is the event that
\[X\geq \frac 34\,2^\ell\implies X-\mu \geq \frac1{12}2^\ell>\frac 1{10}
\sqrt{2^\ell\mu}.\]
Since
$h$ is 5-independent, the $X_x$ are 4-independent, so
by Theorem \ref{thm:4bound}, we get
\[\Pr\left[X \geq \frac 34\,2^\ell\right]\leq 40000/2^{2\ell}=O(1/2^{2\ell}).\]
Thus we can use $P_\ell=O(1/2^{2\ell})$ in Theorem \ref{thm:linpb-anal}, 
and then we get that the expected
operation cost is
\[O\left(1+\sum_{\ell=0}^{\log_2 t}
 2^{\ell} P_\ell\right)=O\left(1+\sum_{\ell=0}^{\log_2 t}
 2^{\ell}/2^{2\ell}\right)=O(1).\]
\end{proof}

\begin{problem} Above we assumed that the range of our
hash function is $[t]$ where $t$ is a power of two. As suggested in the introduction,
we use a hash function based on a degree 4 polynomial over a prime
field $\mathbb Z_p$ where $p\gg 1$, that is, we pick $5$ independent random
coefficients $a_0,\ldots,a_4\in[p]$, and define the hash function
$h':[p]\rightarrow [t]$ by
\[h'(x) = \Big( \big( a_4 x^4 + \cdots + a_1 x + a_0 \big)
        \bmod{p} \Big)\bmod t.\]
Then for any distinct $x_0,\ldots,x_4$, the hash values $h'(x_0),\ldots,h'(x_4)$
are independent. Moreover, we have almost uniformity in the sense
that for any $x\in[p]$ and $y\in [t]$, we have $1/t-1/p<\Pr[h'(x)=y]<1/t+1/p$.

Prove that Theorem \ref{thm:linpb-5} still holds with constant operation time
if $p\geq 24 t$.
\end{problem}

\begin{problem} Assuming full randomness, use Chernoff bounds to
prove that the longest run in the hash table has length $O(\log n)$ with
probability at least $1-1/n^{10}$.

{\em Hint.} You can use Lemma \ref{lem:run} to prove that if there is run of 
length $r\geq 2^{\ell+2}$, then some $\ell$-interval 
is near-full. You can then pick $\ell=C\ln n$ for some large enough constant $C$.
\end{problem}

\begin{problem} Using Chebyshev's inequality, show that with
3-independent hashing, the expected operation time is $O(\log n)$.
\end{problem}

\subsection{Fourth moment and simple tabulation hashing}\label{sec:charhash}
In the preceding analysis we use the 5-independence of the hash
function as follows. First we fix the hash of the query
key. Conditioned on this fixing, we still have 4-independence in the
hashes of the stored keys, and we use this 4-independence to prove the
4th moment bound \req{eq:fourth} on the number stored keys hashing to
any given interval. This was all we needed about the hash function to
conclude that linear probing takes expected constant time per
operation.

P{\v a}tra{\c s}cu and Thorup \cite{patrascu12charhash} 
have proved that something called simple tabulation hashing, that is only
3-independent, within a constant factor provides the same 4th moment 
bound \req{eq:fourth} on the number of stored keys hashing 
to any given interval conditioned on a fixed hash of the query key.
 Linear probing therefore also works
in expected constant time with simple tabulation. This is important
because simple tabulation is 10 times faster than 5-independence implemented 
with a polynomial as in \req{eq:prime}.

Simple tabulation hashing was invented by Zobrist \cite{zobrist70hashing} in 1970 for
chess computers. The basic idea is to view a key $x$ as
consisting of $c$ characters for some constant $c$, e.g., a 32-bit key
could be viewed as consisting of $c=4$ characters of 8 bits.
We initialize $c$ tables $T_1,\ldots,T_c$ mapping characters to 
random hash values that are bit-strings of a certain length.
A key $x=(x_1, ..., x_c)$ is then hashed to $T_1[x_1]\oplus\cdots\oplus T_c[x_c]$
where $\oplus$ denotes bit-wise xor.

\section{The $k$-th moment} 
The 4\th moment bound used above generalizes to any even moment. First
we need
\begin{theorem}\label{thm:k-moment}
Let $X_{0},\ldots,X_{n-1}\in \{0,1\}$ be $k$-wise independent variables for
some (possibly odd) $k\geq 2$. Let $p_i=\Pr[X_i=1]$ and $\sigma_i^2=\Var[X_i]=p_i-p_i^2$. Moreover, let $X=\sum_{i\in[n]} X_i$, $\mu=\E[X]=\sum_{i\in[n]} p_i$,
and $\sigma^2=\Var[X]=\sum_{i\in[n]} \sigma_i^2$. Then
\[\E[(X-\mu)^k]\leq O(\sigma^2+\sigma^k)=O(\mu+\mu^{k/2}).\]
\end{theorem}
\begin{proof} The proof is a simple generalization of the
proof of Theorem \ref{thm:4bound} up to \req{eq:4moment}.
We have
\[(X-\mu)^k=\sum_{i_0,\ldots,i_{k-1}\in[n]} \left((X_{i_0}-p_{i_0})(X_{i_1}-p_{i_1})\cdots (X_{i_{k-1}}-p_{i_{k-1}})\right)\]
By linearity of expectation,
\[\E[(X-\mu)^k]=\sum_{i_0,\ldots,i_{k-1}\in[n]} \E\left[\left((X_{i_0}-p_{i_0})(X_{i_1}-p_{i_1})\cdots (X_{i_{k-1}}-p_{i_{k-1}})\right)\right]\]
We now consider a specific term 
\[\left((X_{i_0}-p_{i_0})(X_{i_1}-p_{i_1})\cdots (X_{i_{k-1}}-p_{i_{k-1}})\right)\]
Let $j_0<j_1<\cdots <j_{c-1}$ be the distinct indices among $i_0,i_1,\ldots,i_{n-1}$, and let $m_h$ be the multiplicity of $j_h$. Then
\begin{align*}&\left((X_{i_0}-p_{i_0})(X_{i_1}-p_{i_1})\cdots (X_{i_{k-1}}-p_{i_{k-1}})\right)\\
&\hspace{3em}=\left((X_{j_0}-p_{j_0})^{m_0}(X_{j_1}-p_{j_1})^{m_1}\cdots (X_{j_{c-1}}-p_{j_{c-1}})^{m_{c-1}}\right).\end{align*}
The product involves at most $k$ different variables so they are all
independent, and therefore 
\begin{align*}
&\E\left[\left((X_{j_0}-p_{j_0})^{m_0}(X_{j_1}-p_{j_1})^{m_1}\cdots (X_{j_{c-1}}-p_{j_{c-1}})^{m_{c-1}}\right)\right]\\
&\hspace{3em}=\E\left[(X_{j_0}-p_{j_0})^{m_0}\right]\;\E\left[(X_{j_1}-p_{j_1})^{m_1}\right]\cdots \E\left[(X_{j_{c-1}}-p_{j_{c-1}})^{m_{c-1}}\right]\end{align*}
Now, for any $i\in[n]$, $\E[X_i-p_i]=0$, so if any multiplicity is 1, the
expected value is zero. We therefore only need to count terms where
all multiplicities $m_h$ are at least $2$. The sum of multiplicities
is $\sum_{h\in [c]}m_h=k$, so we conclude that there are $c\leq k/2$
distinct indices $j_0,\ldots,j_{c-1}$. Now by \req{eq:vari},
\[\E\left[(X_{j_0}-p_{j_0})^{m_0}\right]\;\E\left[(X_{j_1}-p_{j_1})^{m_1}\right]\cdots \E\left[(X_{j_{c-1}}-p_{j_{c-1}})^{m_{c-1}}\right]\leq
 \sigma^2_{j_0}\sigma^2_{j_1}\cdots \sigma^2_{j_{c-1}}.\]
We now want to bound the number tuples 
$(i_0,i_1,\ldots,i_{k-1})$ that have the same $c$ distinct
indices $j_0<j_1<\cdots<j_{c-1}$. A crude upper bound is
that we have $c$ choices for each $i_h$, hence $c^k$ tuples.
We therefore conclude that
\begin{align*}
\E[(X-\mu)^k]&=\sum_{i_0,\ldots,i_{k-1}\in[n]} \E\left[\left((X_{i_0}-p_{i_0})(X_{i_1}-p_{i_1})\cdots (X_{i_{k-1}}-p_{i_{k-1}})\right)\right]\\
&\leq\sum_{c=1}^{\floor{k/2}}\left(
c^k\sum_{0\leq j_0<j_1<\cdots<j_{c-1}<n} \sigma^2_{j_0}\sigma^2_{j_1}\cdots \sigma^2_{j_{c-1}}\right)\\
&\leq\sum_{c=1}^{\floor{k/2}}\left(
\frac{c^k}{c!}\sum_{j_0,j_1,\ldots,j_{c-1}\in[n]} \sigma^2_{j_0}\sigma^2_{j_1}\cdots \sigma^2_{j_{c-1}}\right)\\
&\leq\sum_{c=1}^{\floor{k/2}}\left(
\frac{c^k}{c!}\left(\sum_{j\in[n]}\sigma^2_{j}\right)^c\right)\\
&=\sum_{c=1}^{\floor{k/2}}\left(
\frac{c^k}{c!}\sigma^{2c}\right)\\
&=O\left(\sigma^2+\sigma^{k}\right)=O\left(\mu+\mu^{k/2}\right).\\
\end{align*}
Above we used that $c,k=O(1)$ hence that, e.g., $c^k=O(1)$.
This completes the proof of Theorem~\ref{thm:k-moment}.
\end{proof}
For even moments, we now get a corresponding error probability bound
\begin{corollary}\label{cor:k-moment}
Let $X_{0},\ldots.X_{n-1}\in \{0,1\}$ be $k$-wise independent variables for
some even constant $k\geq 2$. Let $p_i=\Pr[X_i=1]$ and $\sigma_i^2=\Var[X_i]=p_i-p_i^2$. Moreover, let $X=\sum_{i\in[n]} X_i$, $\mu=\E[X]=\sum_{i\in[n]} p_i$,
and $\sigma^2=\Var[X]=\sum_{i\in[n]} \sigma_i^2$. 
If $\mu=\Omega(1)$, then
\[\Pr[|X-\mu|\geq d\sqrt\mu]= O(1/d^{k}).\]
\end{corollary}
\begin{proof} By Theorem \ref{thm:k-moment} and Markov's inequality,
we get 
\begin{align*}
\Pr[|X-\mu|\geq d\sqrt\mu]&=\Pr[(X-\mu)^k\geq d^k\mu^{k/2}]\\
&\leq\frac{\E[(X-\mu)^k]}{d^k\mu^{k/2}}\\
&=\frac{O\left(\mu+\mu^{\floor{k/2}}\right)}{d^k\mu^{k/2}}\\
&=O(1/d^k).
\end{align*}
\end{proof}
\begin{problem} In the proofs of this section, where and why do we need
that (a) $k$ is a constant and (b) that $k$ is even.
\end{problem}

\section{Bloom filters via linear probing}
We will now show how we can reduce the space of a linear probing table
if we are willing to allow for a small chance of false positives, that
is, the table attemps to answer if a query $q$ is in the current
stored set $S$. If it answers ``no'', then $q\not\in S$. If $q\in S$,
then it always answers ``yes''. However, even if $q\not\in S$, then
with some probability $\leq P$, the table may answer ``yes''. Bloom
\cite{bloom70filter} was the first to suggest creating such a filter
using less space than one giving exact answers. Our implementation
here, using linear probing, is completely different. The author
suggested this use of linear probing to various people in the late
90ties, but it was never written down.

To create a filter, we use a universal hash function $s:[u]\fct[2^b]$.
We call $s(x)$ the signature of $x$. The point is that $s(x)$ should
be much smaller than $x$, that is, $b\ll \log_2 u$. The linear probing
array $T$ is now only an array of $t$ signatures. We still use the hash
function $h:[u]\fct[t]$ to start the search for a key in the array. Thus, to check if a key $q$ is positive in
the filter, we look for $s(q)$ among the signatures in $T$ from
location $h(q)$ and onwards until the first empty locatition. If
$s(q)$ is found, we report ``yes''; otherwise ``no''. If we want to
include $q$ to the filter, we only do something if $s(q)$ was not
found. Then we place $s(q)$ it in the first empty location. Our filter
does not support deletion of keys (c.f. Problem \ref{pr:filter-no-deletion}).

\begin{theorem}\label{thm:filter} Assume that the hash function $h$ and the signature
function $s$ are independent, that $h$ is 5-independent, and that $s$
is universal. Then the probability of a false positive on a given key
$q\not\in S$ is $O(1/2^b)$.
\end{theorem}
\begin{proof} The keys from $S$ have been inserted in some given order.
Let us assume that $h$ is fixed. Suppose we inserted the keys
exactly, that is, not just their signatures, and let $X(q)$ be the set
of keys encountered when searching for $q$, that is, $X(q)$ is the set
of keys from $h(q)$ and till the first empty location. Note that $X(q)$
depends only on $h$, not on $s$.

In Problem \ref{pr:monotone-skip} you will argue that if $q$ is a
false positive, then $s(q)=s(x)$ for some $x\in X(q)$. 

For every key $x\in [u]\setminus\{q\}$, by universality of $s$, we have
$\Pr[s(x)=s(q)]\leq 1/2^b$. Since $q\not\in S\supseteq X(q)$, by union, 
$\Pr[\exists x\in X(q): s(x)=s(q)]\leq |X(q)|/2^b$. It follows that the probability that $q$ is a false
positive is bounded by 
\[\sum_{Y\subseteq S}\Pr[X(q)=Y]\cdot |Y|/2^b=\E[|X(q)|]/2^b.\]
By Theorem \ref{thm:linpb-5}, $\E[|X(q)|]=O(1)$ when $h$ is 5-independent.
\end{proof}

\begin{problem}\label{pr:monotone-skip} To complete the proof
of Theorem \ref{thm:filter}, consider a sequence $x_1,\ldots,x_n$ of
distinct keys inserted in an exact linear probing table
(as defined in Section \ref{sec:linpb}). Also, let
$x_{i_1},...,x_{i_m}$ be a subequence of these keys, that is, $1\leq
i_1<i_2<\cdots < i_m\leq n$. The task is to prove any fixed $h:[u]\fct
[t]$ and any fixed $j\in[t]$, that when only the subsequence is
inserted, then the sequence of keys encountered from location $j$ and
till the first empty location is a subsequence of those encountered
when the full sequence is inserted.

{\em Hint.} Using induction on $n$, show that
the above statement is preserved when a new key $x_{n+1}$ is
added. Here $x_{n+1}$ may or may not be part of the subsequence.

The relation to the proof of Theorem \ref{thm:filter} is that when we
insert keys in a filter, we skip keys whose signatures are found as false
postives. This means that only a subsequence of the keys have their 
signatures inserted. When searching for a key $q$ starting
from from location $j=h(q)$, we have thus proved that
we only consider (signatures of) a subset of the set $X(q)$ of keys 
that we would have considered if all keys where inserted. In particular,
this means that if we from $j$ encounter a key $x$ with $s(x)=s(q)$, then
$x\in X(q)$ as required for the proof of Theorem \ref{thm:filter}.
\end{problem}

\begin{problem}\label{pr:filter-no-deletion} Discuss why we canot support
deletions. 
\end{problem}

\begin{problem}\label{pr:hash-signature} What would happen if
we instead used $h(s(x))$ as the hash function to place or find $x$?
What would be the probability of a false positive?
\end{problem}

Sometimes it is faster to generate the hash values and signatures
together so that the pairs $(h(x),s(x))$ are 5-independent while
the hash values and signatures are not necessarily independent of each other.
An example is if we generate a larger hash value, using high-order bits for 
$h(x)$ 
and low-order bits for $s(x)$.
In this situation we get a somewhat weaker bound than that in Theorem \ref{thm:filter}.
\begin{theorem}\label{thm:filter0} Assuming 
that $x\mapsto (h(x),s(x))$ is 5-independent, the probability of a false positive on a given key
$q\not\in S$ is $O(1/2^{2b/3})$.
\end{theorem}
\begin{proof}
Considering the exact insertion of all keys, we consider two
cases. Either (a) there is a run of length at least $2^{b/3}$ around
$h(q)$, or (b) there is no such run.  

For case (a), we use Lemma \ref{lem:query-runs} together with the bound
$P_\ell=O(1/2^{2\ell})$ from the proof of Theorem~\ref{thm:linpb-anal}. We get that the probability of getting a run of
length at least $2^{b/3}$ is bounded by
\[\sum_{\ell=b/3-2}^\infty 12 P_\ell=O(1/2^{2b/3}).\]
We now consider case (b). 
By the statement proved in Problem \ref{pr:monotone-skip}, we
know that any signature $s(x)$ considered is from a key
$x$ from the set $X(q)$ of keys that we would have considered from $j=h(q)$
if all keys were inserted exactly. With
no run of length at least $2^{b/3}$, all
keys in $X(q)$ must hash to $(h(q)-2^{b/3},h(q)+2^{b/3}]$. Thus, if we get a false positive in case (b), it is because there is a key
$x\in S$ with $s(x)=s(q)$ and $h(x)\in
(h(q)-2^{b/3},h(q)+2^{b/3}]$. Since $(h(x),s(x))$ and $(h(q),s(q))$
are independent, the probability that this happens for $x$
is bounded by
$2^{1+b/3}/(t2^b)= O(1/(n2^{2b/3}))$, yeilding $O(1/2^{2b/3})$
when we sum over all $n$ keys in $S$. By union, the probability of a
false positive in case (a) or (b) is bounded by $O(1/2^{2b/3})$, as
desired.
\end{proof}
We note that with the simple tabulation mentioned in Section \ref{sec:charhash}, we can put hash-signature pairs
as concatenated bit strings in the character tables $T_1,\ldots,T_c$.
Then $(h(x),s(x))=T_1[x_1]\oplus\cdots\oplus T_c[x_c]$. The nice
thing here is that with simple tabulation hashing, the output bits
are all completely independent, which means that Theorem \ref{thm:filter}
applies even though we generate the hash-signature pairs together.

%\bibliographystyle{plain}
%\bibliography{general}

\end{document}